\documentclass[]{jot}
\usepackage[T1]{fontenc}

\usepackage{xcolor}
\usepackage{amssymb}
\usepackage{mathpartir}
\usepackage{stmaryrd}
\usepackage{amsthm}
\usepackage{thmtools, thm-restate}

\usepackage{etoolbox}
\AtBeginEnvironment{array}{}

\newtheorem{lemma}{Lemma}

\jotdetails{
    volume=25,      
    number=3,       
    articleno=a5,   
    year=2025,      
    license=ccby    
}
\articletype{regular} 

\title{VeriFast's separation logic: a logic without laters for modular verification of fine-grained concurrent programs}

\runningtitle{VeriFast's separation logic}

\author{Bart Jacobs
}

\runningauthor{Bart Jacobs}

\affil{KU Leuven, Department of Computer Science, DistriNet Research Group, Leuven, Belgium}


\definecolor{ghost}{HTML}{CC6600}
\definecolor{darkgreen}{rgb}{0.0,0.5,0.0}

\newcommand{\gmapsto}{\mapsto_\mathsf{g}}
\newcommand{\annot}[1]{{\color{blue} #1}}
\newcommand{\ghost}[1]{{\color{ghost} #1}}
\newcommand{\internal}[1]{{\color{magenta} #1}}
\newcommand{\comment}[1]{{\color{darkgreen} #1}}

\newcommand{\llbrace}{\{\hspace{-3pt}[}
\newcommand{\rrbrace}{]\hspace{-3pt}\}}

\begin{abstract}
VeriFast is one of the leading tools for semi-automated modular formal program verification.
A central feature of VeriFast is its support for \emph{higher-order ghost code}, which enables
its support for expressively specifying fine-grained concurrent modules, without the need for the
\emph{later} modality. We present the first formalization and soundness proof
for this aspect of VeriFast's logic, and we compare it both to Iris, a state-of-the-art logic for fine-grained concurrency which features the later modality, as well as to some recent proposals for Iris-like reasoning without the later modality.
\end{abstract}

\acknowledgement{We thank Justus Fasse for proofreading. This research is partially funded by the Research Fund KU Leuven, and by the Cybersecurity Research Program Flanders.}

\begin{document}
\maketitle
\urlstyle{rm}

\section{Introduction}

VeriFast \cite{fvf} is one of the leading tools for semi-automated modular formal verification of single-threaded and multithreaded C, Java, and Rust programs. It symbolically executes each function/method of the program, using a separation logic \cite{seplogic-csl01,seplogic-lics02} representation of memory. It requires programs to be \emph{annotated} with function/method preconditions and postconditions and loop invariants, as well as \emph{ghost declarations}, such as definitions of separation logic predicates that specify the layout of data structures, and \emph{ghost commands} for folding and unfolding predicates as well as invoking \emph{lemma functions}, functions consisting entirely of ghost code. For expressive specification of fine-grained concurrent modules, it supports \emph{higher-order ghost code}, in the form of \emph{lemma function pointers} and \emph{lemma function pointer type assertions}. While the general ideas underlying this specification approach have been described earlier \cite{popl11}, as have some examples of their use for solving verification challenges \cite{verifythis2012,verifythis2016}, in this paper we present the first formalization and soundness proof for this aspect of VeriFast's logic. We define the programming language and introduce the running example in \S\ref{sec:lang}, define the syntax of annotations in \S\ref{sec:annot}, formalize the program logic implemented by VeriFast's symbolic execution algorithm in \S\ref{sec:verif}, and prove its soundness in \S\ref{sec:soundness}. We discuss related work in \S\ref{sec:relatedwork}.

\section{Programming language}\label{sec:lang}

In order to focus on the complexities of the logic rather than those of the programming language, we present VeriFast's separation logic in the context of a trivial concurrent programming language whose syntax is given in Fig.~\ref{fig:program-syntax} and whose small-step operational semantics is given in Fig.~\ref{fig:program-steps}. An example program that allocates a memory cell, increments it twice in parallel, and then asserts that the cell's value equals two is shown in Fig.~\ref{fig:example}.

\begin{figure}
$$\begin{array}{l}
\mathbf{let}\ \mathsf{x} = \mathbf{cons}(0)\ \mathbf{in}\\
(\ \mathbf{FAA}(\mathsf{x}, 1)\ ||\ \mathbf{FAA}(\mathsf{x}, 1)\ );\\
\mathbf{let}\ \mathsf{v} = {*}\mathsf{x}\ \mathbf{in}\\
\mathbf{assert}\ \mathsf{v} = 2
\end{array}$$
\caption{An example program. $\mathbf{cons}(0)$ allocates a memory cell, initializes it to 0, and returns its address. The $\mathbf{FAA}$ command performs a sequentially consistent atomic fetch-and-add operation. $c_1\,||\, c_2$ is the parallel composition of commands $c_1$ and $c_2$. $*\mathsf{x}$ returns the value stored at address $\mathsf{x}$.}\label{fig:example}
\end{figure}

\begin{figure}
$$\begin{array}{r @{\;} @{\;} l}
& z \in \mathbb{Z}, x \in \mathcal{X}\\
e ::= & z\ |\ x\\
i ::= & \mathbf{cons}(e)\ |\ \mathbf{FAA}(e, e)\ |\ {*}e\ |\ \mathbf{assert}\ e = e\\
c ::= & e\ |\ i\ |\ \mathbf{let}\ x = c\ \mathbf{in}\ c\ |\ (c\;||\;c)
\end{array}$$
\caption{Syntax of the expressions $e$, instructions $i$, and commands $c$ of the programming language. We assume a set $\mathcal{X}$ of program variable names. $c; c'$ is a shorthand for $\mathbf{let}\ \_ = c\ \mathbf{in}\ c'$, where $\_$ is a designated element of $\mathcal{X}$}\label{fig:program-syntax}
\end{figure}

\begin{figure}
\begin{mathpar}
\inferrule{
\ell \notin \mathrm{dom}\,h
}{
(h, \mathbf{cons}(v)) \rightarrow (h[\ell := v], \ell)
}
\and
\inferrule{
\ell \in \mathrm{dom}\,h
}{
(h, \mathbf{FAA}(\ell, z) \rightarrow (h[\ell := h(\ell) + z], h(\ell))
}
\and
\inferrule{
\ell \in \mathrm{dom}\,h
}{
(h, {*}\ell) \rightarrow (h, h(\ell))
}
\and
(h, \mathbf{assert}\ v = v) \rightarrow (h, 0)
\and
(h, \mathbf{let}\ x = v\ \mathbf{in}\ c) \rightarrow (h, c[v/x])
\and
\inferrule{
(h, c) \rightarrow (h', c')
}{
(h, \mathbf{let}\ x = c\ \mathbf{in}\ c'') \rightarrow (h', \mathbf{let}\ x = c'\ \mathbf{in}\ c'')
}
\and
\inferrule{
(h, c) \rightarrow (h', c')
}{
(h, (c\;||\;c'')) \rightarrow (h', (c'\;||\;c''))
}
\and
\inferrule{
(h, c) \rightarrow (h', c')
}{
(h, (c''\;||\;c)) \rightarrow (h', (c''\;||\;c'))
}
\and
(h, v\;||\;v') \rightarrow (h, 0)
\end{mathpar}
\caption{Small-step operational semantics of the programming language}\label{fig:program-steps}
\end{figure}

We define the multiset of threads of a command $c$ as follows:
$$\mathsf{thrds}(c) = \left\{\begin{array}{l l}
\mathsf{thrds}(c_1) & \textrm{if $c = \mathbf{let}\ x = c_1\ \mathbf{in}\ c_2$}\\
\mathsf{thrds}(c_1) \uplus \mathsf{thrds}(c_2) & \textrm{if $c = (c_1\;||\;c_2)$}\\
\llbrace c\rrbrace & \textrm{otherwise}
\end{array}\right.$$

We say a configuration $(h, c)$ is \emph{reducible} if it can make a step:
$$\inferrule{
(h, c) \rightarrow (h', c')
}{
\mathsf{red}\,(h, c)
}$$
We say a configuration is \emph{finished} if its command is a value.
$$\mathsf{finished}\,(h, z)$$
We say a configuration is \emph{okay} if each thread is either reducible or finished.
$$\inferrule{
\forall c_\mathsf{t} \in \mathsf{thrds}(c).\;\mathsf{finished}\,(h, c_\mathsf{t}) \lor \mathsf{red}\,(h, c_\mathsf{t})
}{
\mathsf{ok}\,(h, c)
}$$
We say a configuration is \emph{safe} if each configuration reachable from it is okay.
$$\inferrule{
\forall h', c'.\;(h, c) \rightarrow^* (h', c') \Rightarrow \mathsf{ok}\,(h', c')
}{
\mathsf{safe}\,(h, c)
}$$
We say a program $c$ is safe if $(\emptyset, c)$ is safe. The goal of the logic that we present here is to prove that a given program is safe. This implies that it does not access unallocated memory and that there are no assertion failures.\footnote{In fact, the logic also proves that there are no data races, but for simplicity we do not consider data races here.}

\section{Annotated programs}\label{sec:annot}

\begin{figure*}
$$\begin{array}{r @{\;} l r}
& t \in \mathcal{T} & \textrm{lemma type names}\\
& p \in \mathcal{P} & \textrm{predicate constructor names}\\
& g \in \mathcal{G} & \textrm{ghost variable names}\\
& \pi \in \mathbb{R}^+ & \textrm{fractions}\\
\textrm{ghost values $V$} ::= & z\ |\ (V, V)\ |\ ()\ |\ \{\overline{V}\}\\
| & p(\overline{V}) & \textrm{predicate values}\\
| & \lambda \overline{g}.\;G & \textrm{lemma values}\\
\textrm{ghost expressions $E$} ::= & V\ |\ x\ |\ g\ |\ E + E\\
| & p(\overline{E}) & \textrm{predicate constructor applications}\\
| & (E, E)\ |\ () & \textrm{pair expressions, empty tuple}\\
| & \emptyset\ |\ \{E\}\ |\ E \cup E\ |\ E \setminus E & \textrm{set expressions}\\
\textrm{assertions $a$} ::= & [\pi]E \mapsto E & \textrm{points-to assertions}\\
| & [\pi]E \gmapsto E & \textrm{ghost cell points-to assertions}\\
| & E() & \textrm{predicate assertions}\\
| & [\pi]\mathbf{atomic\_space}(E, E) & \textrm{atomic space assertions}\\
| & E : t(\overline{E}) & \textrm{lemma type assertions}\\
| & \exists g.\;a\\
| & \mathbf{atomic\_spaces}(E) & \textrm{atomic spaces assertions}\\
| & \internal{\mathbf{heap}(E)} & \textrm{heap chunk assertions}\\
| & a * a & \textrm{separating conjunctions}\\
\mathit{gdecl} ::= & \multicolumn{2}{@{} l @{}}{\mathbf{lem\_type}\ t(\overline{g}) = \mathbf{lem}(\overline{g})\ \mathbf{forall}\ \overline{g}\ \mathbf{req}\ a\ \mathbf{ens}\ a}\\
| & \mathbf{pred\_ctor}\ p(\overline{g})() = a\\
I ::= & E(\overline{E})\\
| & \mathbf{gcons}(E)\ |\ {*}E \leftarrow_\mathsf{g} E\\
| & \multicolumn{2}{@{} l @{}}{\mathbf{open\_atomic\_space}(E, E)\ |\ \mathbf{close\_atomic\_space}(E, E)}\\
| & \internal{E \leftarrow_\mathsf{h} E} & \textrm{heap chunk update}\\
G ::= & I\ |\ \mathbf{glet}_\mathsf{i}\ g = G\ \mathbf{in}\ G\\
C ::= & \multicolumn{2}{@{} l @{}}{G\ |\ \mathbf{produce\_lem\_ptr\_chunk}\ t(\overline{E})(\overline{g})\ \{\ G\ \}}\\
| & \multicolumn{2}{@{} l @{}}{\mathbf{create\_atomic\_space}(E, E)\ |\ \mathbf{destroy\_atomic\_space}(E, E)}\\
\hat{c} ::= & \multicolumn{2}{@{} l @{}}{e\ |\ i\ |\ \mathbf{let}\ x = \hat{c}\ \mathbf{in}\ \hat{c}\ |\ \hat{c}\ ||\ \hat{c}\ |\ \mathbf{glet}\ g = C\ \mathbf{in}\ \hat{c}}
\end{array}$$
\caption{Syntax of ghost declarations $\mathit{gdecl}$, ghost instructions $I$, inner ghost commands $G$, outer ghost commands $C$ (collectively called \emph{ghost commands}), and annotated commands $\hat{c}$. Heap chunk assertions and heap chunk update commands are \internal{\emph{internal}}; they are not accepted by VeriFast in source code and are introduced here only for the sake of the soundness proof.}\label{fig:annot-syntax}
\end{figure*}

\begin{figure}
$$\begin{array}{l}
\ghost{\begin{array}{@{} l @{}}
\mathbf{pred\_ctor}\ \mathsf{Inv}(\mathsf{x}, \mathsf{g1}, \mathsf{g2})() =\\
\quad \exists \mathsf{v1}, \mathsf{v2}.\;[1/2]\mathsf{g1} \gmapsto \mathsf{v1} * [1/2]\mathsf{g2} \gmapsto \mathsf{v2} * \mathsf{x} \mapsto \mathsf{v1} + \mathsf{v2}\\
\mathbf{pred\_ctor}\ \mathsf{pre1}(\mathsf{x}, \mathsf{g1}, \mathsf{g2})() =\\
quad [1/2]\mathbf{atomic\_space}(\mathsf{Nx}, \mathsf{Inv}(\mathsf{x}, \mathsf{g1}, \mathsf{g2})) * [1/2]\mathsf{g1} \gmapsto 0\\
\mathbf{pred\_ctor}\ \mathsf{post1}(\mathsf{x}, \mathsf{g1}, \mathsf{g2})() =\\
\quad [1/2]\mathbf{atomic\_space}(\mathsf{Nx}, \mathsf{Inv}(\mathsf{x}, \mathsf{g1}, \mathsf{g2})) * [1/2]\mathsf{g1} \gmapsto 1\\
\mathbf{pred\_ctor}\ \mathsf{pre2}(\mathsf{x}, \mathsf{g1}, \mathsf{g2})() =\\
\quad [1/2]\mathbf{atomic\_space}(\mathsf{Nx}, \mathsf{Inv}(\mathsf{x}, \mathsf{g1}, \mathsf{g2})) * [1/2]\mathsf{g2} \gmapsto 0\\
\mathbf{pred\_ctor}\ \mathsf{post2}(\mathsf{x}, \mathsf{g1}, \mathsf{g2})() =\\
\quad [1/2]\mathbf{atomic\_space}(\mathsf{Nx}, \mathsf{Inv}(\mathsf{x}, \mathsf{g1}, \mathsf{g2})) * [1/2]\mathsf{g2} \gmapsto 1
\end{array}}\\
\\
\mathbf{let}\ \mathsf{x} = \mathbf{cons}(0)\ \mathbf{in}\\
\ghost{\mathbf{glet}\ \mathsf{g1} = \mathbf{gcons}(0)\ \mathbf{in}}\\
\ghost{\mathbf{glet}\ \mathsf{g2} = \mathbf{gcons}(0)\ \mathbf{in}}\\
\ghost{\mathbf{create\_atomic\_space}(\mathsf{Nx}, \mathsf{Inv}(\mathsf{x}, \mathsf{g1}, \mathsf{g2}));}\\
(\\
\quad \ghost{\begin{array}{@{} l @{}}
\mathbf{produce\_lem\_ptr\_chunk}\\
\quad \mathsf{FAA\_ghop}(\mathsf{x}, 1, \mathsf{pre1}(\mathsf{x}, \mathsf{g1}, \mathsf{g2}), \mathsf{post1}(\mathsf{x}, \mathsf{g1}, \mathsf{g2}))(\mathsf{op})\ \{\\
\quad \mathbf{open\_atomic\_space}(\mathsf{Nx}, \mathsf{Inv}(\mathsf{x}, \mathsf{g1}, \mathsf{g2}));_\mathsf{i}\\
\quad \mathsf{op}();_\mathsf{i}\\
\quad *\mathsf{g1} \leftarrow_\mathsf{g} 1;_\mathsf{i}\\
\quad \mathbf{close\_atomic\_space}(\mathsf{Nx}, \mathsf{Inv}(\mathsf{x}, \mathsf{g1}, \mathsf{g2}))\\
\};
\end{array}}\\
\quad \mathbf{FAA}(\mathsf{x}, 1)\\
||\\
\quad \ghost{\begin{array}{@{} l @{}}
\mathbf{produce\_lem\_ptr\_chunk}\\
\quad \mathsf{FAA\_ghop}(\mathsf{x}, 1, \mathsf{pre2}(\mathsf{x}, \mathsf{g1}, \mathsf{g2}), \mathsf{post2}(\mathsf{x}, \mathsf{g1}, \mathsf{g2}))(\mathsf{op})\ \{\\
\quad \mathbf{open\_atomic\_space}(\mathsf{Nx}, \mathsf{Inv}(\mathsf{x}, \mathsf{g1}, \mathsf{g2}));_\mathsf{i}\\
\quad \mathsf{op}();_\mathsf{i}\\
\quad *\mathsf{g2} \leftarrow_\mathsf{g} 1;_\mathsf{i}\\
\quad \mathbf{close\_atomic\_space}(\mathsf{Nx}, \mathsf{Inv}(\mathsf{x}, \mathsf{g1}, \mathsf{g2}))\\
\};
\end{array}}\\
\quad \mathbf{FAA}(\mathsf{x}, 1)\\
);\\
\ghost{\mathbf{destroy\_atomic\_space}(\mathsf{Nx}, \mathsf{Inv}(\mathsf{x}, \mathsf{g1}, \mathsf{g2}));}\\
\mathbf{let}\ \mathsf{v} = {*}\mathsf{x}\ \mathbf{in}\\
\mathbf{assert}\ \mathsf{v} = 2
\end{array}$$
\caption{VeriFast proof of the example program. $\mathsf{Nx} \triangleq ()$.}\label{fig:example-proof}
\end{figure}

\begin{figure}
$$\ghost{\begin{array}{@{} l @{}}
\mathbf{lem\_type}\ \mathsf{FAA\_op}(\mathsf{x}, \mathsf{n}, \mathsf{P}, \mathsf{Q}) = \mathbf{lem}()\\
\quad \mathbf{forall}\ \mathsf{v}\\
\quad \mathbf{req}\ \mathsf{x} \mapsto \mathsf{v} * \mathsf{P}()\\
\quad \mathbf{ens}\ \mathsf{x} \mapsto \mathsf{v} + \mathsf{n} * \mathsf{Q}()\\
\mathbf{lem\_type}\ \mathsf{FAA\_ghop}(\mathsf{x}, \mathsf{n}, \mathsf{pre}, \mathsf{post}) = \mathbf{lem}(\mathsf{op})\\
\quad \mathbf{forall}\ \mathsf{P}, \mathsf{Q}\\
\quad \mathbf{req}\ \mathbf{atomic\_spaces}(\emptyset) * \mathsf{op} : \mathsf{FAA\_op}(\mathsf{x}, \mathsf{n}, \mathsf{P}, \mathsf{Q}) * \mathsf{P}()\\
\quad\quad *\; \mathsf{pre}()\\
\quad \mathbf{ens}\ \mathbf{atomic\_spaces}(\emptyset) * \mathsf{op} : \mathsf{FAA\_op}(\mathsf{x}, \mathsf{n}, \mathsf{P}, \mathsf{Q}) * \mathsf{Q}()\\
\quad\quad *\; \mathsf{post}()\\
\internal{\mathbf{pred\_ctor}\ \mathsf{heap\_}(\mathsf{h})() = \mathbf{heap}(\mathsf{h})}\\
\end{array}}$$
\caption{The ghost prelude (built-in ghost declarations). The declaration of $\mathsf{heap\_}$ is \internal{\emph{internal}}. It is not meant to be used in annotated programs; it is introduced here only for the sake of the soundness proof.}\label{fig:prelude}
\end{figure}

When verifying a program with VeriFast, the user must first insert \emph{annotations}, specifically \emph{ghost declarations} and \emph{ghost commands}, to obtain an \emph{annotated program}. The syntax of ghost declarations and ghost commands is shown in Fig.~\ref{fig:annot-syntax}. An annotated version of the example program is shown in Fig.~\ref{fig:example-proof}. An annotated program may refer to ghost constructs declared in the \emph{VeriFast prelude}, shown in Fig.~\ref{fig:prelude}.

There are two kinds of ghost declarations: \emph{lemma type declarations} and \emph{predicate constructor declarations}. These give meaning to \emph{lemma type names} $t \in \mathcal{T}$ and \emph{predicate constructor names} $p \in \mathcal{P}$. Conceptually, a lemma type is a predicate over a \emph{lemma value} $\lambda \overline{g}.\;C$, a parameterized ghost command. A predicate constructor is a named, parameterized assertion. Applying a predicate constructor to an argument list produces a \emph{predicate value} $p(\overline{V})$.

Besides integers, lemma values, and predicate values, ghost values may be pairs of ghost values, unit values $()$, and finite sets of ghost values.

Resources may be shared among threads using \emph{atomic spaces} (analogous to Iris \emph{invariants} \cite{iris1,iris-ground-up}). An atomic space is (non-uniquely) identified by a \emph{name} (any ghost value) and an \emph{invariant} (a predicate value) (but there may be multiple atomic spaces with the same name and invariant at any given time). At any point in time, ownership of the stock of logical resources in the system is distributed among the threads and the atomic spaces. That is, at any point, each logical resource is owned either by exactly one thread or by exactly one atomic space, or has been leaked irrecoverably. (More precisely, given that fractional resources are supported, the bundles of resources owned by the threads and the atomic spaces sum up to a logical heap that contains each physical points-to chunk only once and each $\mathbf{atomic\_space}$ chunk only as many times as there are atomic spaces with that name and invariant, etc.) Creating an atomic space transfers a bundle of resources satisfying the atomic space's invariant from the creating thread to the newly created atomic space. Opening an atomic space transfers the resources owned by the atomic space to the opening thread; closing an atomic space again transfers a bundle of resources satisfying the atomic space's invariant from the closing thread to the atomic space. Destroying an atomic space transfers ownership of the resources owned by the atomic space to the destroying thread. To destroy an atomic space, the destroying thread must have full ownership of the atomic space. To open it, only partial ownership is required. (To close it, no ownership is required. If no such atomic space exists, the resources are leaked.) To prevent the same atomic space from being opened when it is already open, the set of opened atomic spaces is tracked using an $\mathbf{atomic\_spaces}(S)$ chunk, where $S$ is a set of the name-invariant pairs of the atomic spaces that are currently open.\footnote{This means it is not possible to open two atomic spaces with the same name-invariant pair at the same time, even if multiple such atomic spaces exist.}

Lemma type assertions $V : t(\overline{V})$ assert that a given lemma value $V$ is of a given lemma type $t$, applied to a given lemma type argument list $\overline{V}$. Such assertions are \emph{linear}. To call a lemma, a full lemma type chunk for that lemma must be available, and it becomes unavailable for the duration of the call. A lemma type chunk is produced by the $\mathbf{produce\_lem\_ptr\_chunk}$ ghost command. Since that command is not allowed inside lemmas, the stock of lemma type chunks in the system only decreases as the lemma call stack grows; absence of infinite lemma recursion follows trivially.\footnote{This is a simplification with respect to the actual VeriFast tool, which does support production of lemma type chunks inside lemmas, using a variant of the $\mathbf{produce\_lem\_ptr\_chunk}$ syntax that additionally takes a block of ghost code. The chunk is available only until the end of that block. Now, suppose there is an infinite lemma call stack. Since the program text contains only finitely many $\mathbf{produce\_lem\_ptr\_chunk}$ commands, among the lemmas that appear infinitely often in that call stack, there is one that is syntactically maximal, i.e. that is not itself contained within another lemma that also appears infinitely often. It follows that from some point on, the call stack contains no lemmas bigger than this maximal one. Since a lemma type chunk for a given lemma can only be produced by a bigger lemma (since the latter's body must contain a $\mathbf{produce\_lem\_ptr\_chunk}$ command producing the former's), the stock of lemma type chunks for this maximal lemma will, from that point on, only decrease, which leads to a contradiction. (Note: for measuring the size of a lemma, the size of contained lemma \emph{values} is not taken into account. It follows that substitution of values for ghost variables never affects the size of a lemma.)}

Intermediate results produced by ghost commands can be stored in \emph{ghost variables}, which are like program variables except that they are in a separate namespace and can therefore never hide a program variable.\footnote{In the actual VeriFast tool, they are in the same namespace, but VeriFast checks that real code never uses a ghost variable.} To facilitate reasoning about concurrent programs, annotated programs can furthermore allocate \emph{ghost cells}; these are like physical memory locations except that they are allocated in a separate \emph{ghost heap} and mutated using separate \emph{ghost cell mutation commands}.

Points-to chunks, ghost points-to chunks, and atomic spaces can be owned \emph{fractionally}, which allows them to be shared temporarily or permanently among multiple threads. A \emph{fractional chunk} has a \emph{coefficient} which is a positive real number.

\section{Verification of annotated programs}\label{sec:verif}

In this section we formalize the program logic implemented by VeriFast's symbolic execution algorithm. We abstract over the mechanics of symbolic execution, the essence of which is described in Featherweight VeriFast \cite{fvf}. In particular, the tool generally requires $\mathbf{open}$ and $\mathbf{close}$ ghost commands to unfold and fold predicates. Instead, here we use \emph{semantic assertions}; predicates are fully unfolded during the interpretation of syntactic assertions as semantic assertions.

Core to VeriFast's verification approach is the concept of a \emph{chunk} $\alpha$:
$$\begin{array}{r @{\;} l}
\alpha ::= & V \mapsto V\ |\ V \mapsto_\mathsf{g} V\ |\ \mathbf{atomic\_space}(V, V)\ |\ V : t(\overline{V})\\
& |\ \mathbf{atomic\_spaces}(V)\ |\ \mathbf{heap}(V)
\end{array}$$
A \emph{logical heap} $H$ is a function from chunks to nonnegative real numbers:
$$H \in \mathit{LogicalHeaps} = \mathit{Chunks} \rightarrow \mathbb{R}^+$$

We say a logical heap is \emph{weakly consistent}, denoted $\mathsf{wok}\;H$ if no points-to chunk or ghost points-to chunk is present with a coefficient greater than 1, and no two (fractions of) points-to chunks or two (fractions of) ghost points-to chunks are present with the same left-hand side (address) but a different right-hand side (stored value).

We define \emph{satisfaction} of an assertion $a$ by a logical heap $H$, denoted $H \vDash a$, inductively as follows:
\begin{mathpar}
\inferrule{
H(\alpha) \ge \pi
}{
H \vDash [\pi]\alpha
}
\and
\inferrule{
\mathbf{pred\_ctor}\ p(\overline{g})() = a\\
|\overline{V}| = |\overline{g}|\\
H \vDash a[\overline{V}/\overline{g}]
}{
H \vDash p(\overline{V})()
}
\and
\inferrule{
H \vDash a[V/g]
}{
H \vDash \exists g.\;a
}
\and
\inferrule{
H \vDash a\\
H' \vDash a'
}{
H + H' \vDash a * a'
}
\end{mathpar}

A \emph{semantic assertion} is a set of logical heaps. We define the \emph{interpretation} $\llbracket a\rrbracket$ of an assertion as a semantic assertion as $\llbracket a\rrbracket = \{H\ |\ H\vDash a\}$.

We define \emph{correctness} of an annotated command or ghost command $\dot{c}$ with respect to a precondition $P$ and a postcondition $Q$ (both semantic assertions), denoted $\{P\}\ \dot{c}\ \{Q\}$, inductively in Fig.~\ref{fig:correctness}. We define implication of semantic assertions as follows:
$$P \Rightarrow Q \triangleq \forall H \in P.\;\mathsf{wok}\,H \Rightarrow H \in Q$$

\begin{figure*}
\begin{mathpar}
\{\mathsf{True}\}\ \mathbf{cons}(V)\ \{\mathsf{res} \mapsto V\}
\and
\{[\pi]\ell \mapsto V\}\ {*}\ell\ \{[\pi]\ell \mapsto V \land \mathsf{res} = V\}
\and
\{\ell \mapsto V\}\ \ell \leftarrow V'\ \{\ell \mapsto V'\}
\and
\inferrule{
\{P\}\ \hat{c}\ \{R\}\\
\forall v.\;\{R[v/\mathsf{res}]\}\ \hat{c}'[v/x]\ \{Q\}
}{
\{P\}\ \mathbf{let}\ x = \hat{c}\ \mathbf{in}\ \hat{c}'\ \{Q\}
}
\and
\begin{array}{l}
\{V : \mathsf{FAA\_ghop}(\ell, z, V', V'') * \llbracket V'()\rrbracket\}\cr
\mathbf{FAA}(\ell, z)\cr
\{V : \mathsf{FAA\_ghop}(\ell, z, V', V'') * \llbracket V''()\rrbracket\}
\end{array}
\and
\inferrule{
\{P\}\ \hat{c}\ \{Q\}\\
\{P'\}\ \hat{c}'\ \{Q'\}
}{
\{P * P'\}\ \hat{c}\ ||\ \hat{c}'\ \{Q * Q'\}
}
\and
\{\mathsf{True}\}\ \mathbf{gcons}(V)\ \{\mathsf{res} \gmapsto V\}
\and
\{\ell \gmapsto V\}\ \ell \leftarrow_\mathsf{g} V'\ \{\ell \gmapsto V'\}
\and
\{\llbracket V'()\rrbracket\}\ \mathbf{create\_atomic\_space}(V, V')\ \{\mathbf{atomic\_space}(V, V')\}
\and
\inferrule{
(V, V') \notin S
}{
\begin{array}{l}
\{\mathbf{atomic\_spaces}(S) * [\pi]\mathbf{atomic\_space}(V, V')\}\cr
\mathbf{open\_atomic\_space}(V, V')\cr
\{\mathbf{atomic\_spaces}(S \cup \{(V, V')\}) * [\pi]\mathbf{atomic\_space}(V, V') * \llbracket V'()\rrbracket\}
\end{array}
}
\and
\begin{array}{l}
\{\mathbf{atomic\_spaces}(S) * \llbracket V'()\rrbracket\}\cr
\mathbf{close\_atomic\_space}(V, V')\cr
\{\mathbf{atomic\_spaces}(S \setminus \{(V, V')\})\}
\end{array}
\and
\{\mathbf{atomic\_space}(V, V')\}\ \mathbf{destroy\_atomic\_space}(V, V')\ \{\llbracket V'()\rrbracket\}
\and
\inferrule{
\mathbf{lem\_type}\ t(\overline{g}) = \mathbf{lem}(\overline{g}')\ \mathbf{req}\ a\ \mathbf{ens}\ a'\\
|\overline{V}| = |\overline{g}|\\
|\overline{g}''| = |\overline{g}'|\\
\forall \overline{V}'.\;|\overline{V}'| = |\overline{g}'| \Rightarrow \{\llbracket a[\overline{V}/\overline{g},\overline{V}'/\overline{g}']\rrbracket\}\ G[\overline{V}'/\overline{g}'']\ \{\llbracket a'[\overline{V}/\overline{g},\overline{V}'/\overline{g}']\rrbracket\}
}{
\{\mathsf{True}\}\ \mathbf{produce\_lem\_ptr\_chunk}\ t(\overline{V})(\overline{g}'')\ \{\ G\ \}\ \{\mathsf{res} : t(\overline{V})\}
}
\and
\inferrule{
\mathbf{lem\_type}\ t(\overline{g}) = \mathbf{lem}(\overline{g}')\ \mathbf{req}\ a\ \mathbf{ens}\ a'\\
|\overline{V}'| = |\overline{g}'|
}{
\{V : t(\overline{V}) * \llbracket a[\overline{V}/\overline{g},\overline{V}'/\overline{g}']\rrbracket\}\ V(\overline{V}')\ \{V : t(\overline{V}) * \llbracket a'[\overline{V}/\overline{g},\overline{V}'/\overline{g}']\rrbracket\}
}
\and
\{\mathbf{heap}(h) * \ell \mapsto \_\}\ \ell \leftarrow_\mathsf{h} v\ \{\mathsf{heap}(h[\ell := v]) * \ell \mapsto v\}
\and
\inferrule{
\{P\}\ \dot{c}\ \{Q\}
}{
\{P * R\}\ \dot{c}\ \{Q * R\}
}
\and
\inferrule{
\forall V.\;\{P[V/g]\}\ \dot{c}\ \{Q\}
}{
\{\exists g.\;P\}\ \dot{c}\ \{Q\}
}
\and
\inferrule{
P \Rightarrow P'\\
\{P'\}\ \dot{c}\ \{Q\}\\
Q \Rightarrow Q'
}{
\{P\}\ \dot{c}\ \{Q'\}
}
\end{mathpar}
\caption{Correctness of annotated commands and ghost commands. We use $\dot{c}$ to range over both annotated commands and ghost commands.}\label{fig:correctness}
\end{figure*}

Note: nesting $\mathbf{produce\_lem\_ptr\_chunk}$ commands is not allowed.

A correctness proof outline for the example annotated program is shown in Fig.~\ref{fig:example-outline}.

\begin{figure*}
$$\begin{array}{l}
\annot{\{\mathbf{emp}\}}\\
\mathbf{let}\ \mathsf{x} = \mathbf{cons}(0)\ \mathbf{in}\ \ghost{\mathbf{glet}\ \mathsf{g1} = \mathbf{gcons}(0)\ \mathbf{in}}\ \ghost{\mathbf{glet}\ \mathsf{g2} = \mathbf{gcons}(0)\ \mathbf{in}}\\
\annot{\{\mathsf{x} \mapsto 0 * \mathsf{g1} \gmapsto 0 * \mathsf{g2} \gmapsto 0\}}\\
\comment{\mathbf{close}\ \mathsf{Inv}(\mathsf{x}, \mathsf{g1}, \mathsf{g2})();}\\
\annot{\{\mathsf{Inv}(\mathsf{x}, \mathsf{g1}, \mathsf{g2})() * [1/2]\mathsf{g1} \gmapsto 0 * [1/2]\mathsf{g2} \gmapsto 0\}}\\
\ghost{\mathbf{create\_atomic\_space}(\mathsf{Nx}, \mathsf{Inv}(\mathsf{x}, \mathsf{g1}, \mathsf{g2}));}\\
\annot{\{\mathbf{atomic\_space}(\mathsf{Nx}, \mathsf{Inv}(\mathsf{x}, \mathsf{g1}, \mathsf{g2})) * [1/2]\mathsf{g1} \gmapsto 0 * [1/2]\mathsf{g2} \gmapsto 0\}}\\
(\\
\quad \annot{\{[1/2]\mathbf{atomic\_space}(\mathsf{Nx}, \mathsf{Inv}(\mathsf{x}, \mathsf{g1}, \mathsf{g2})) * [1/2]\mathsf{g1} \gmapsto 0\}}\\
\quad \ghost{\begin{array}{@{} l @{}}
\mathbf{glet}\ \mathsf{lem} = \mathbf{produce\_lem\_ptr\_chunk}\ \mathsf{FAA\_ghop}(\mathsf{x}, 1, \mathsf{pre1}(\mathsf{x}, \mathsf{g1}, \mathsf{g2}), \mathsf{post1}(\mathsf{x}, \mathsf{g1}, \mathsf{g2}))(\mathsf{op})\ \{\\
\quad \annot{\textrm{For all $\mathsf{P}, \mathsf{Q},$}}\\
\quad \annot{\{\mathbf{atomic\_spaces}(\emptyset) * \mathsf{op} : \mathsf{FAA\_op}(\mathsf{x}, 1, \mathsf{P}, \mathsf{Q}) * \mathsf{P}() * \mathsf{pre1}(\mathsf{x}, \mathsf{g1}, \mathsf{g2})()\}}\\
\quad \comment{\mathbf{open}\ \mathsf{pre1}(\mathsf{x}, \mathsf{g1}, \mathsf{g2})();}\\
\quad \annot{\left\{\begin{array}{l}
\mathbf{atomic\_spaces}(\emptyset) * \mathsf{op} : \mathsf{FAA\_op}(\mathsf{x}, 1, \mathsf{P}, \mathsf{Q}) * \mathsf{P}() * {}\\
{}[1/2]\mathbf{atomic\_space}(\mathsf{Nx}, \mathsf{Inv}(\mathsf{x}, \mathsf{g1}, \mathsf{g2})) * [1/2]\mathsf{g1} \gmapsto 0
\end{array}\right\}}\\
\quad \mathbf{open\_atomic\_space}(\mathsf{Nx}, \mathsf{Inv}(\mathsf{x}, \mathsf{g1}, \mathsf{g2}));\ \comment{\mathbf{open}\ \mathsf{Inv}(\mathsf{x}, \mathsf{g1}, \mathsf{g2})();}\\
\quad \annot{\left\{\begin{array}{l}
\exists \mathsf{v2}.\;\mathbf{atomic\_spaces}(\{(\mathsf{Nx}, \mathsf{Inv}(\mathsf{x}, \mathsf{g1}, \mathsf{g2}))\}) * \mathsf{op} : \mathsf{FAA\_op}(\mathsf{x}, 1, \mathsf{P}, \mathsf{Q}) * \mathsf{P}() * {}\\
{}[1/2]\mathbf{atomic\_space}(\mathsf{Nx}, \mathsf{Inv}(\mathsf{x}, \mathsf{g1}, \mathsf{g2})) * \mathsf{g1} \gmapsto 0 * [1/2]\mathsf{g2} \gmapsto \mathsf{v2} * \mathsf{x} \mapsto \mathsf{v2}
\end{array}\right\}}\\
\quad \annot{\textrm{For all $\mathsf{v2}$,}}\\
\quad \annot{\left\{\begin{array}{l}
\mathbf{atomic\_spaces}(\{(\mathsf{Nx}, \mathsf{Inv}(\mathsf{x}, \mathsf{g1}, \mathsf{g2}))\}) * \mathsf{op} : \mathsf{FAA\_op}(\mathsf{x}, 1, \mathsf{P}, \mathsf{Q}) * \mathsf{P}() * {}\\
{}[1/2]\mathbf{atomic\_space}(\mathsf{Nx}, \mathsf{Inv}(\mathsf{x}, \mathsf{g1}, \mathsf{g2})) * \mathsf{g1} \gmapsto 0 * [1/2]\mathsf{g2} \gmapsto \mathsf{v2} * \mathsf{x} \mapsto \mathsf{v2}
\end{array}\right\}}\\
\quad \mathsf{op}();\\
\quad \annot{\left\{\begin{array}{l}
\mathbf{atomic\_spaces}(\{(\mathsf{Nx}, \mathsf{Inv}(\mathsf{x}, \mathsf{g1}, \mathsf{g2}))\}) * \mathsf{op} : \mathsf{FAA\_op}(\mathsf{x}, 1, \mathsf{P}, \mathsf{Q}) * \mathsf{Q}() * {}\\
{}[1/2]\mathbf{atomic\_space}(\mathsf{Nx}, \mathsf{Inv}(\mathsf{x}, \mathsf{g1}, \mathsf{g2})) * \mathsf{g1} \gmapsto 0 * [1/2]\mathsf{g2} \gmapsto \mathsf{v2} * \mathsf{x} \mapsto 1 + \mathsf{v2}
\end{array}\right\}}\\
\quad *\mathsf{g1} \leftarrow_\mathsf{g} 1;\\
\quad \annot{\left\{\begin{array}{l}
\mathbf{atomic\_spaces}(\{(\mathsf{Nx}, \mathsf{Inv}(\mathsf{x}, \mathsf{g1}, \mathsf{g2}))\}) * \mathsf{op} : \mathsf{FAA\_op}(\mathsf{x}, 1, \mathsf{P}, \mathsf{Q}) * \mathsf{Q}() * {}\\
{}[1/2]\mathbf{atomic\_space}(\mathsf{Nx}, \mathsf{Inv}(\mathsf{x}, \mathsf{g1}, \mathsf{g2})) * \mathsf{g1} \gmapsto 1 * [1/2]\mathsf{g2} \gmapsto \mathsf{v2} * \mathsf{x} \mapsto 1 + \mathsf{v2}
\end{array}\right\}}\\
\quad \comment{\mathbf{close}\ \mathsf{Inv}(\mathsf{x}, \mathsf{g1}, \mathsf{g2})();}\ \mathbf{close\_atomic\_space}(\mathsf{Nx}, \mathsf{Inv}(\mathsf{x}, \mathsf{g1}, \mathsf{g2}));\\
\quad \annot{\left\{\begin{array}{l}
\mathbf{atomic\_spaces}(\emptyset) * \mathsf{op} : \mathsf{FAA\_op}(\mathsf{x}, 1, \mathsf{P}, \mathsf{Q}) * \mathsf{Q}() * {}\\
{}[1/2]\mathbf{atomic\_space}(\mathsf{Nx}, \mathsf{Inv}(\mathsf{x}, \mathsf{g1}, \mathsf{g2})) * [1/2]\mathsf{g1} \gmapsto 1
\end{array}\right\}}\\
\quad \comment{\mathbf{close}\ \mathsf{post1}(\mathsf{x}, \mathsf{g1}, \mathsf{g2})()}\\
\quad \annot{\{\mathbf{atomic\_spaces}(\emptyset) * \mathsf{op} : \mathsf{FAA\_op}(\mathsf{x}, 1, \mathsf{P}, \mathsf{Q}) * \mathsf{Q}() * \mathsf{post1}(\mathsf{x}, \mathsf{g1}, \mathsf{g2})()\}}\\
\}\ \mathbf{in}\\
\end{array}}\\
\quad \annot{\{[1/2]\mathbf{atomic\_space}(\mathsf{Nx}, \mathsf{Inv}(\mathsf{x}, \mathsf{g1}, \mathsf{g2})) * [1/2]\mathsf{g1} \gmapsto 0 * \mathsf{lem} : \mathsf{FAA\_ghop}(\mathsf{x}, 1, \mathsf{pre1}(\mathsf{x}, \mathsf{g1}, \mathsf{g2}), \mathsf{post1}(\mathsf{x}, \mathsf{g1}, \mathsf{g2}))\}}\\
\quad \comment{\mathbf{close}\ \mathsf{pre1}(\mathsf{x}, \mathsf{g1}, \mathsf{g2})();}\\
\quad \annot{\{\mathsf{pre1}(\mathsf{x}, \mathsf{g1}, \mathsf{g2})() * \mathsf{lem} : \mathsf{FAA\_ghop}(\mathsf{x}, 1, \mathsf{pre1}(\mathsf{x}, \mathsf{g1}, \mathsf{g2}), \mathsf{post1}(\mathsf{x}, \mathsf{g1}, \mathsf{g2}))\}}\\
\quad \mathbf{FAA}(\mathsf{x}, 1);\\
\quad \annot{\{\mathsf{post1}(\mathsf{x}, \mathsf{g1}, \mathsf{g2})() * \mathsf{lem} : \mathsf{FAA\_ghop}(\mathsf{x}, 1, \mathsf{pre1}(\mathsf{x}, \mathsf{g1}, \mathsf{g2}), \mathsf{post1}(\mathsf{x}, \mathsf{g1}, \mathsf{g2}))\}}\\
\quad \comment{\mathbf{open}\ \mathsf{post1}(\mathsf{x}, \mathsf{g1}, \mathsf{g2})()}\\
\quad \annot{\{[1/2]\mathbf{atomic\_space}(\mathsf{Nx}, \mathsf{Inv}(\mathsf{x}, \mathsf{g1}, \mathsf{g2})) * [1/2]\mathsf{g1} \gmapsto 1\}}\\
||\\
\quad \dots\\
);\\
\annot{\{\mathbf{atomic\_space}(\mathsf{Nx}, \mathsf{Inv}(\mathsf{x}, \mathsf{g1}, \mathsf{g2})) * [1/2]\mathsf{g1} \gmapsto 1 * [1/2]\mathsf{g2} \gmapsto 1\}}\\
\ghost{\mathbf{destroy\_atomic\_space}(\mathsf{Nx}, \mathsf{Inv}(\mathsf{x}, \mathsf{g1}, \mathsf{g2}));}\ \comment{\mathbf{open}\ \mathsf{Inv}(\mathsf{x}, \mathsf{g1}, \mathsf{g2})();}\\
\annot{\{\mathsf{g1} \gmapsto 1 * \mathsf{g2} \gmapsto 1 * \mathsf{x} \mapsto 2\}}\\
\mathbf{let}\ \mathsf{v} = {*}\mathsf{x}\ \mathbf{in}\\
\mathbf{assert}\ \mathsf{v} = 2
\end{array}$$
\caption{Proof outline for the example proof}\label{fig:example-outline}
\end{figure*}

We say an annotated program $\hat{c}$ is \emph{correct} if $\{\mathsf{True}\}\ \hat{c}\ \{\mathsf{True}\}$.

We define the erasure of an annotated command $\hat{c}$ to a command $c = \mathsf{erasure}(\hat{c})$ as follows:
$$\begin{array}{r @{\;} l}
\mathsf{erasure}(c) = & c\\
\mathsf{erasure}(\mathbf{let}\ x = \hat{c}\ \mathbf{in}\ \hat{c}') = & \mathbf{let}\ x = \mathsf{erasure}(\hat{c})\ \mathbf{in}\ \mathsf{erasure}(\hat{c}')\\
\mathsf{erasure}(\hat{c}\;||\;\hat{c}') = & \mathsf{erasure}(\hat{c})\;||\;\mathsf{erasure}(\hat{c}')\\
\mathsf{erasure}(\mathbf{glet}\ g = C\ \mathbf{in}\ \hat{c}) = & \mathsf{erasure}(\hat{c})
\end{array}$$

\begin{restatable}{theorem}{maintheorem}\label{thm:main}
If an annotated program $\hat{c}$ is correct, then its erasure $\mathsf{erasure}(\hat{c})$ is safe.
\end{restatable}

\section{Soundness}\label{sec:soundness}

We say a logical heap is \emph{strongly consistent}, denoted $\mathsf{sok}\;H$, if, for every $V : t(\overline{V})$ such that $H(V : t(\overline{V})) > 0$, we have
that $V$ semantically is of type $t(\overline{V})$, denoted $\vDash V : t(\overline{V})$, defined as follows:
$$\inferrule{
\mathbf{lem\_type}\ t(\overline{g}')(\overline{g}'')\ \mathbf{req}\ a\ \mathbf{ens}\ a''\\
|\overline{V}| = |\overline{g}'|\\
|\overline{g}| = |\overline{g}''|\\
\forall \overline{V}'.\;\begin{array}{c}
|\overline{V}'| = |\overline{g}|\; \Rightarrow\cr
\{\llbracket a[\overline{V}/\overline{g}', \overline{V}'/\overline{g''}]\rrbracket\}\ G[\overline{V}'/\overline{g}]\ \{\llbracket a'[\overline{V}/\overline{g}', \overline{V}'/\overline{g''}]\rrbracket\}
\end{array}
}{
\vDash \lambda \overline{g}.\;G : t(\overline{V})
}$$

A \emph{ghost heap} $\hat{h}$ is a partial function from integers to ghost values.

An \emph{atomic spaces bag} $A$ is a multiset of pairs $((V, V), H)$ of name-invariant pairs and logical heaps, such that for each element $((\_, V), H)$ we have $H \vDash V()$. We define the atomic space chunks $\mathsf{chunks}(A)$ and the atomic spaces total owned heap $\mathsf{heap}(A)$ as follows:
$$\begin{array}{r l}
\mathsf{chunks}(A) = & \llbrace \mathbf{atomic\_space}(V, V')\ |\ ((V, V'), \_) \in A\rrbrace\\
\mathsf{heap}(A) = & \biguplus_{(\_, H) \in A} H
\end{array}$$

A \emph{stock of lemma type chunks} $\Sigma$ is a multiset of $(V, t, \overline{V})$ tuples. We say such a stock is \emph{consistent} if for each $(V, t, \overline{V})$ in $\Sigma$, $V$ is semantically of type $t(\overline{V})$.

We say a heap $h$ and logical heap $H$ are \emph{consistent}, denoted $h \sim H$, if there exists a ghost heap $\hat{h}$, an atomic spaces bag $A$, and a consistent stock of lemma type chunks $\Sigma$ such that $h + \hat{h} + \mathsf{chunks}(A) + \Sigma \ge \mathsf{heap}(A) + H$, where a heap is interpreted as a set of $\mapsto$ chunks and a ghost heap is interpreted as a set of $\gmapsto$ chunks. Notice: if $h \sim H$, it follows that $H$ is strongly consistent.

We define the \emph{weakest precondition} for $n$ steps of a command $c$ with respect to postcondition $Q$, denoted $\mathsf{wp}_n(c, Q)$, as the semantic assertion that is true for a logical heap $H$ if either $c$ is a value and $H \in Q$ or $n = 0$ or for each heap $h$ and \emph{frame} $H'$ such that $h \sim H + H'$, all threads of $c$ are either finished or reducible and for each step that $(h, c)$ can make to some configuration $(h', c')$, there exists a logical heap $H''$ such that $h' \sim H'' + H'$ and $H''$ satisfies the weakest precondition of $c'$ with respect to $Q$ for $n - 1$ steps:
$$\begin{array}{l}
H \in \mathsf{wp}_n(c, Q) \Leftrightarrow\\
\quad \mathsf{finished}(\emptyset, c) \land H \in Q \lor n = 0\; \lor\\
\quad \forall h, H'.\;h \sim H + H' \Rightarrow (h, c)\,\mathsf{ok} \land {}\\
\quad\quad \forall h', c'.\;(h, c) \rightarrow (h', c')\; \Rightarrow\\
\quad\quad\quad \exists H''.\;h' \sim H'' + H' \land H'' \in \mathsf{wp}_{n - 1}(c', Q)
\end{array}$$

We define the atomic space chunks $\mathsf{chunks}(S)$ for a set $S$ of opened atomic spaces as follows:
$$\mathsf{chunks}(S) = \llbrace \mathbf{atomic\_space}(V, V')\ |\ (V, V') \in S\rrbrace$$

We say a logical heap $H$ is \emph{self-consistent with depth bound $k$}, denoted $H\,\mathsf{ok}_k$, if there exists a heap $h$, a ghost heap $\hat{h}$, an atomic spaces bag $A$, a set of opened atomic spaces $S$, and a consistent stock of lemma type chunks $\Sigma$ of size at most $k$ such that $\llbrace\mathbf{heap}(h)\rrbrace + h + \hat{h} + \mathsf{chunks}(A) + \mathsf{chunks}(S) + \llbrace\mathbf{atomic\_spaces}(S)\rrbrace + \Sigma \ge \mathsf{heap}(A) + H$, where a heap is interpreted as a set of $\mapsto$ chunks and a ghost heap is interpreted as a set of $\gmapsto$ chunks. Notice: if $H\,\mathsf{ok}_k$, it follows that $H$ is strongly consistent.

Notice that $h \sim H$ if and only if $\exists k, (H + \llbrace\mathbf{heap}(h), \mathbf{atomic\_spaces}(\emptyset)\rrbrace)\,\mathsf{ok}_k$.

\begin{lemma}[Soundness of inner ghost command correctness]\label{lem:G-sound}
$$\begin{array}{l}
\{P\}\ G\ \{Q\} \land H \in P \land (H + H')\,\mathsf{ok}_k\; \Rightarrow\\
\quad \exists H'' \in Q.\; (H'' + H')\,\mathsf{ok}_k
\end{array}$$
\end{lemma}
\begin{proof}
By induction on $k$ and nested induction on the size of $G$. The outer induction hypothesis is used to deal with lemma calls.
\end{proof}

\begin{lemma}\label{lem:correct-wp}
If an annotated command $\hat{c}$ is correct with respect to precondition $P$ and postcondition $Q$, then, for all $n$, $P$ implies the weakest precondition of the erasure of $\hat{c}$ with respect to $Q$ for $n$ steps:
$$\{P\}\ \hat{c}\ \{Q\} \Rightarrow \forall n.\;P \Rightarrow \mathsf{wp}_n(\mathsf{erasure}(\hat{c}), Q)$$
\end{lemma}
\begin{proof}
By induction on the derivation of the correctness judgment.
The most interesting case is $\hat{c} = \mathbf{FAA}(\ell, z)$. Fix an $n$ and a logical heap $H \in P$. Fix a heap $h$, a ghost heap $\hat{h}$, an atomic spaces bag $A$, a consistent stock of lemma type chunks $\Sigma$, and a frame $H_\mathsf{F}$ such that $h + \hat{h} + \mathsf{chunks}(A) + \Sigma = \mathsf{heap}(A) + H + H_\mathsf{F}$. By $H \in P$ and $H$ strongly consistent we can fix a $g$, a $G$ and an $H'$ such that $H = \llbrace \lambda g.\;G : \mathsf{FAA\_ghop}(\ell, z, V_\mathsf{pre}, V_\mathsf{post})\rrbrace + H'$ and $H' \vDash V_\mathsf{pre}()$. By strong consistency of $H$, we have $\forall \mathit{op}, V_\mathsf{P}, V_\mathsf{Q}.\;\{\mathbf{atomic\_spaces}(\emptyset) * \mathit{op} : \mathsf{FAA\_op}(\ell, z, V_\mathsf{P}, V_\mathsf{Q}) * \llbracket V_\mathsf{P}()\rrbracket * \llbracket V_\mathsf{pre}()\rrbracket\}\ G[\mathit{op}/g]\ \{\mathbf{atomic\_spaces}(\emptyset) * \mathit{op} : \mathsf{FAA\_op}(\ell, z, V_\mathsf{P}, V_\mathsf{Q}) * \llbracket V_\mathsf{Q}()\rrbracket * \llbracket V_\mathsf{post}()\rrbracket\}$. We take $\mathit{op} = \lambda.\;\ell \leftarrow_\mathsf{h} h(\ell) + z$, $V_\mathsf{P} = \mathsf{heap\_}(h)$, and $V_\mathsf{Q} = \mathsf{heap\_}(h[\ell := h(\ell) + z])$. We have that semantically, $\mathit{op}$ is of type $\mathsf{FAA\_op}(\ell, z, V_\mathsf{P}, V_\mathsf{Q})$, so $\Sigma' = \Sigma - \llbrace \lambda g.\;G : \mathsf{FAA\_ghop}(\ell, z, V_\mathsf{pre}, V_\mathsf{post})\rrbrace + \llbrace \mathsf{op} : \mathsf{FAA\_op}(\ell, z, V_\mathsf{P}, V_\mathsf{Q})\rrbrace$ is consistent. We apply Lemma~\ref{lem:G-sound} to $G$ using $H' + \llbrace \mathbf{atomic\_spaces}(\emptyset), \mathsf{op} : \mathsf{FAA\_op}(\ell, z, V_\mathsf{P}, V_\mathsf{Q}), \mathbf{heap}(h)\rrbrace$ for $H$, $H_\mathsf{F}$ for $H'$ and the size of $\Sigma'$ for $k$ to obtain that there exists an $H'' \in \llbracket V_\mathsf{post}()\rrbracket$ such that $(H'' + \llbrace \mathbf{atomic\_spaces}(\emptyset), \mathsf{op} : \mathsf{FAA\_op}(\ell, z, V_\mathsf{P}, V_\mathsf{Q}), \mathbf{heap}(h[\ell:=h(\ell) + z])\rrbrace + H_\mathsf{F})\,\mathsf{ok}_k$ and therefore $\llbrace \lambda g.\;G : \mathsf{FAA\_ghop}(\ell, z, V_\mathsf{pre}, V_\mathsf{post})\rrbrace + H'' \in Q$ and $h[\ell := h(\ell) + z] \sim H'' + \llbrace \lambda g.\;G : \mathsf{FAA\_ghop}(\ell, z, V_\mathsf{pre}, V_\mathsf{post})\rrbrace + H_\mathsf{F}$.
\end{proof}

\begin{lemma}\label{lem:wp-okay}
If $h \sim H$ and $H \in \mathsf{wp}_n(c, \mathsf{True})$, then any configuration reached by $(h, c)$  in at most $n$ steps is okay.
\end{lemma}
\begin{proof}
By induction on $n$.
\end{proof}

\maintheorem*
\begin{proof}
We first apply Lemma~\ref{lem:correct-wp}. Then, since the empty heap is consistent with the empty logical heap, we can finish the proof by applying Lemma~\ref{lem:wp-okay} with $H = \emptyset$.
\end{proof}

\section{Related work and discussion}\label{sec:relatedwork}

\paragraph{Iris} In contrast to the state-of-the-art logic for fine-grained concurrency verification Iris \cite{iris1,iris-ground-up}, the presented logic does not require the \emph{later} modality. This is because atomic space invariants are stored in the logical heap in a \emph{syntactic} form, rather than as propositions over logical heaps. As a result, no recursive domain equations are involved.

A downside of our approach compared to Iris, however, is that our logic does not directly support separating implications (a.k.a.~magic wands), viewshifts, or other logical connectives in which some operand assertions appear in \emph{non-positive} positions, i.e.~whose truth is not monotonic in the truth of some of the operand assertions. This is because we define the meaning of predicate values using a least fixpoint construction.

We recover the functionality of separating implications and viewshifts to some extent by means of lemma values, with the major limitation that lemma type assertions are \emph{linear}, which makes them more awkward to work with than the Iris constructs, although in practice this has not hindered us significantly so far; in fact, while we do vaguely remember encountering cases where this was inconvenient (or worse), we have trouble recalling the specific circumstances.

Having said that, we use VeriFast as a tool for verifying particular programs, not for metatheory development. It is very likely that the limitations of our logic would become prohibitive if we attempted to replicate deep metatheory developments such as RustBelt's lifetime logic \cite{rustbelt} in VeriFast. We do, however, make use of the \emph{results} of such developments in VeriFast, through axiomatisation. The soundness of such axiomatisations, however, is a nontrivial question. While our axiomatisation of the lifetime logic \emph{appears} sound, it is future work to build a formal argument of that, perhaps by connecting a Rocq mechanisation of the development of the present paper with that of the lifetime logic. Very important related work in this regard is Nola, which proposes a version of the lifetime logic, machine-checked in Rocq, without the later modality (see below).

\paragraph{Nola, Lilo} VeriFast's logic is by no means the only later-less logic for fine-grained concurrency. Nola \cite{matsushita-phd,nola} is an Iris library that generalizes Iris's support for invariants (the Iris construct analogous to our atomic spaces) by parameterizing it over the type $\mathit{Fml}$ of \emph{formulas} that describe the contents of invariants, and the \emph{semantics $\llbracket\,\rrbracket$ of formulas}, that maps a formula to an Iris proposition. Classical Iris invariants are obtained by taking $\mathit{Fml} = {\blacktriangleright}\mathsf{iProp}$, the type of Iris propositions but guarded by the later functor, and $\llbracket \mathsf{next}(P)\rrbracket = \triangleright P$. However, an alternative is to take as $\mathit{Fml}$ a type of \emph{syntactic} separation logic formulae. For most syntactic constructs, the semantics can be defined without the need for the later modality, thus enabling later-less invariant reasoning for a wide class of invariants. \citet{nola} were even able to implement a later-less version of RustBelt's lifetime logic this way. Furthermore, they show how to handle an even wider class of invariants without laters by applying \emph{stratification}, where formulae at a layer $k + 1$ can quantify over formulae at layer $k$. Lilo \cite{lilo} applies Nola's idea of stratification to enable later-less invariant reasoning in a logic for verifying termination of busy-waiting programs under fair scheduling. Lilo was used to build the first modular total correctness proof of an elimination stack. Nola and Lilo are fully mechanized in Rocq.

\bibliography{verifasts-logic}

\section*{About the author}
\shortbio{Bart Jacobs}{is an associate professor at the DistriNet research group at the department of Computer Science at KU Leuven (Belgium). His main research interest is in modular formal verification of concurrent programs. \authorcontact[https://distrinet.cs.kuleuven.be/people/BartJacobs/]{bart.jacobs@kuleuven.be}}

\end{document}